\documentclass[sn-mathphys,Numbered]{sn-jnl}

\usepackage[graphicx]{realboxes}%
\usepackage{multirow}%
\usepackage{amsmath,amssymb,amsfonts}%
\usepackage{amsthm}%
\usepackage{mathrsfs}%
\usepackage[title]{appendix}%
\usepackage{xcolor}%
\usepackage{textcomp}%
\usepackage{manyfoot}%
\usepackage{booktabs}%
\usepackage{algorithm}%
\usepackage{algorithmicx}%
\usepackage{algpseudocode}%
\usepackage{listings}%
\usepackage[mathscr]{euscript}
\usepackage{makecell}
\usepackage{array}
\usepackage{mathrsfs}

\newtheorem{theorem}{Theorem}
\newtheorem{example}{Example}%
\newtheorem{remark}{Remark}%

\newtheorem{definition}{Definition}%

\newtheorem{lemma}{{{\textit{Lemma}}}}
\newtheorem{corollary}{{{{\textit{Corollary}}}}}

\newtheorem{claim}{{{\textit{Claim}}}}

\newtheorem{case}{{{\textit{Case}}}}

\begin{document}

\title[Construction of SNC-ZCZ Sequences with Flexible Support]{Construction of Spectrally-Null-Constrained Zero-Correlation Zone Sequences with Flexible Support}


\author[1]{Nishant Kumar}\email{nishant\_1921ma08@iitp.ac.in}

\author[2]{Palash Sarkar}\email{palash.sarkar@uib.no}

\author*[3]{Sudhan Majhi}\email{smajhi@iisc.ac.in}

\affil[1]{Department of Mathematics, IIT Patna, Patna, 801106, Bihar, India}

\affil[2]{Department of Informatics, Selmar Center, University of Bergen, N-5020 Bergen, Norway}

\affil*[3]{Department of Electrical Communication Engineering, IISC, Bangalore, 560012, Karnataka, India}


\abstract{In recent years, traditional zero-correlation zone (ZCZ) sequences are being studied due to support interference-free quasi-synchronous code division multiple access (QS-CDMA) systems. However, in cognitive radio (CR) network, it is desirable to design ZCZ sequences having spectral null constraint (SNC) property to achieve low spectral density profile. This paper focuses on the construction of SNC-ZCZ sequences having flexible support, where support refers to a collection of indices corresponding to non-zero entries in the sequence. The proposed SNC-ZCZ sequences reduce to traditional ZCZ sequences when the size of the support becomes equal to the length of the sequence. To obtain ZCZ sequences, we first propose construction of traditional/SNC-Complete complementary codes (SNC-CCCs) using a class of extended Boolean functions (EBFs). With the help of this class we propose another class of EBFs that generates asymptotically optimal traditional/SNC-ZCZ sequences of prime-power lengths with respect to Tang-Fan-Matsufuzi bound. Furthermore, a relation between the second-order cosets of first-order generalized Reed-Muller code and the proposed traditional ZCZ sequences is also established.}
 



\maketitle

\section{Introduction}\label{sec1}
The construction of sequences possessing favorable correlation properties assumes a crucial significance in the realm of wireless communications and radar sensing. Such sequences find widespread utility across various applications, including but not limited to channel estimation, synchronization, spread spectrum communications, surveillance, and ranging \cite{fan1996sequence}. 
The polyphase sequences having zero-correlation zone (ZCZ) properties have been found suitable for application in quasi-synchronous-code division multiple access (QS-CDMA) system to provide interference free environment. A ZCZ is characterized by a region of zero auto- and cross-correlations centered around the in-phase timing position. These traditional ZCZ sequences (having non-zero entries) assume that there are continuous spectral bands available, so the energy of the sequence can be distributed across all the carriers within that spectral band. Since there's a limited and valuable amount of available radio spectrum for wireless communication and radar sensing, researchers are trying to find better ways to use it efficiently \cite{zhao2007survey}. This emerging concept is referred to as cognitive radio (CR) \cite{haykin2005cognitive} or cognitive radar \cite{haykin2006cognitive}. A significant challenge in cognitive radio (CR) systems is achieving interference-free multiple access communications with a low spectral density profile and robust anti-jamming capability. In pursuit of this goal, a potential strategy involves the utilization of spread-spectrum CDMA techniques, employing sequences possessing the ZCZ property. Exploiting the correlation attribute of ZCZ sequences, a QS-CDMA system utilizing ZCZ sequences as its spreading codes can achieve interference-free performance, contingent on the condition that all received signals align with the ZCZ \cite{suehiro1994signal}. In the literature, ZCZ sequences were first proposed in \cite{fan1996sequence}, and have been well developed by many researchers \cite{ChWu,LiGuPa,TaFa,Th,ZhDz,ZhPa}. Nonetheless, traditional ZCZ sequences encounter a challenge when applied directly in CR systems. This arises due to their design presumption, which typically assumes the presence of the entire spectral band, rather than the specific non-contiguous spectral bands dictated by the spectrum hole constraint within a CR system. Imposing this constraint for spectral nulling jeopardizes the ZCZ property within a traditional ZCZ sequence set, leading to potential loss or degradation of its intended benefits. In this paper, sequences with constraints of spectral nulls are referred
to as spectrally-null-constrained (SNC) sequences. The term "support" pertains to a set of indices that correspond to elements within the sequence that have non-zero values. Therefore, SNC-Sequences can also be defined as the sequences whose size of the support is not equal to the length of the sequence. Additionally, if the size of the support equals to the length of the sequence, then the sequence is referred to as a traditional sequence.
 
In the realm of systematic SNC sequence construction, the pioneering work of Tsai \emph{et al.} \cite{tsai2011lower} introduced the initial SNC-ZCZ family by utilizing a comb-like sequence distribution in the frequency domain. In 2018, Liu \emph{et al.} \cite{liu2018spectrally} achieved SNC sequences with minimal autocorrelation values through carefully selected ternary frequency-domain duals possessing zero periodic autocorrelation sidelobes. A recent contribution by Tian \emph{et al.} \cite{tian2020family} presented an analytical construction of SNC sequences in single-channel polyphase configurations. The work by Li \emph{et al.} \cite{li2022spectrally} proposed the construction of both SNC-ZCZ sequences and SNC-Z-periodic complementary set (SNC-ZPCS) sequences. Furthermore, in a recent study, Ye \emph{et al.} \cite{ye2022new} introduced a method for constructing SNC-ZCZ sequences based on a circular Florentine rectangle and an interleaving approach.

This paper is driven by the limited availability of design approaches for SNC-ZCZ sequences in the existing literature. Due to their mathematical properties, Extended Boolean functions (EBFs) are widely recognized in the literature for their ability to generate sequences with flexible parameters. This motivates us to study EBFs as a way to overcome the current limitations in achieving flexible parameters for traditional/SNC-ZCZ sequences. In this study, we first propose a class of EBFs to generate traditional/SNC-CCC which further been extended to another class of EBFs to achieve traditional/SNC-ZCZ sequences of parameter $(p^t,(p-1)p^n,p^{n+t+1})$. The proposed ZCZ sequences has length in the form of prime-power and alphabet set $\mathbb{Z}_p$. Furthermore, the proposed traditional/SNC-ZCZ sequences are asymptotically optimal as the obtained parameters attain the Tang-Fan-Matsufuzi bound asymptotically. Moreover, we also establish a relationship between traditional ZCZ sequences and GRM codes. The remaining sections of this paper are structured as follows.

Section 2 presents some basic definitions and notations which will be used throughout the paper. The proposed constructions of traditional/SNC-CCC and traditional/SNC-ZCZ sequences are provided in Section 3 and Section 4 respectively. Section 5 establishes the relation between proposed traditional ZCZ sequences and GRM codes. In Section 6 we compare the proposed construction with some existing construction. Lastly, we conclude the paper in Section 7.

\section{Preliminary}\label{sec2}

\subsection{Extended Boolean Functions (EBF) and Corresponding Sequences}
For any $b \geq 2$, an extended Boolean function (EBF) $f(\mathbf{x})$ is defined as a function $f: \mathbb{Z}_b^t \rightarrow$ $\mathbb{Z}_b$, where $\mathbf{x}=\left(x_1, x_2, \ldots, x_t\right)$ is the $b$-ary representation of the integer $x=\sum_{k=1}^t x_k b^{k-1}$. For a given EBF $f$, let $f_i=f\left(i_1, i_2, \ldots, i_t\right)$, and define the $\mathbb{Z}_b$-valued sequence $\xi(f)$ of length $b^t$ associated with $f$ as
\begin{equation}
    \xi(f)=\left[f_0, f_1, \ldots, f_{b^t-1}\right],
\end{equation}
and complex valued sequence associated with $f$ as
\begin{equation}
    \psi(f)=[\omega^{f_0}, \omega^{f_1}, \hdots, \omega^{f_{b^{\color{black}t}-1}}],\label{eq:12}
\end{equation}
where $\omega=\exp(2\pi\iota/b)$.

\subsection{Definition and Correlation Functions}
Let $\boldsymbol{a}=\left(a_0, a_1, \cdots, a_{N-1}\right)$ and $\boldsymbol{b}=\left(b_0, b_1, \cdots, b_{N-1}\right)$ be two complex-valued sequences of equal length $L$. For an integer $u$, define
\begin{equation}\label{equ:cross}
\gamma(\boldsymbol{a}, \boldsymbol{b})(u)=
\begin{cases}
\sum_{i=0}^{L-1-u}a_{i}b^{*}_{(i+u)}, & 0 \leq u < L,
\\
\sum_{i=0}^{L+u-1} a_{(i-u)}b^{*}_{i}, & -L< u < 0.  
\end{cases}
\end{equation}
The functions $\gamma(\boldsymbol{a}, \boldsymbol{b})$ is called aperiodic cross-correlation function (ACCF) of $\boldsymbol{a}$ and $\boldsymbol{b}$. Moreover, if $\boldsymbol{a}=\mathbf{b}$, then this function is called an aperiodic auto-correlation function (AACF).
\par {\color{black}Further, the periodic cross-correlation function (PCCF) of $\boldsymbol{a}$ and $\boldsymbol{b}$ is defined as
\begin{equation}\label{equ:periodic}
\phi(\boldsymbol{a}, \boldsymbol{b})(u)=
\begin{cases}
\sum_{i=0}^{L}a_{i}b^{*}_{(i+u)\mod~L}, & 0 \leq u < L,
\\
\sum_{i=0}^{L} a_{(i-u)\mod~L}b^{*}_{i}, & -L< u < 0.  
\end{cases}
\end{equation}
When $\boldsymbol{a}=\boldsymbol{b}$, then this function is called periodic auto-correlation function (PACF). We can also define PCCF in terms of ACCF as}
\begin{equation} \label{eq7} 
{\color{black}\phi}(\boldsymbol{a},\boldsymbol{b})(u)=\gamma(\boldsymbol{a},\boldsymbol{b})(u)+\gamma^*(\boldsymbol{b},\boldsymbol{a})(L-u).
\end{equation} 

For a sequence $\boldsymbol{a}$, its polynomial representation can be given by 
\begin{equation}
    p_{\boldsymbol{a}}(z)=a_0+a_1z+\hdots+a_{N-1}z^{N-1},
\end{equation}
where $z\in \{e^{j2\pi t:~0\leq t<1}\}$ is a complex number. Further. the following property is easy to derive
\begin{equation}
    \arrowvert p_{\boldsymbol{a}}(z)\arrowvert^2=~ \gamma(\boldsymbol{a})(0)+\sum_{u=1}^{N-1}{(\gamma_{\boldsymbol{a}}(u)z^{-u}+\gamma^*_{\boldsymbol{a}}(u)z^{u})}. 
\end{equation}
 \begin{definition}
Let $\mathbf{C}=\{\mathbf{C}_0,\mathbf{C}_1, \hdots ,\mathbf{C}_{K-1}\}$ be a collection of $K$ matrices (codes) of order $M\times L$ . Define $\mathbf{C_{\mu}}=[(\mathbf{a}_{0}^{\mu})^T,(\mathbf{a}_{1}^{\mu})^T,\hdots,(\mathbf{a}_{M-1}^{\mu})^T]^T$, where $\mathbf{a}_\nu^\mu$ ($0\leq \nu \leq M-1,0 \leq \mu \leq K-1$) is the $\nu$th row sequence or $\nu$th constituent sequence and $T$ is transpose operator. Then the ACCF of two codes $\mathbf{C_{\mu_1}},\mathbf{C_{\mu_2}}\in \textbf{C}$ is defined as
\begin{equation}
    \gamma(\mathbf{C_{\mu_1}},\mathbf{C_{\mu_2}})(u)=\sum_{\nu=0}^{M-1}{\gamma(\mathbf{a}_{\nu}^{\mu_1}, \mathbf{a}_{\nu}^{\mu_2})(u)}.
\end{equation}
\end{definition}
\begin{definition}
  Let \textbf{C} be a code set defined as above and satisfying the following correlation properties,
   \begin{equation}
       \gamma(\mathbf{C_{\mu_1}},\mathbf{C_{\mu_2}})(u)=
      \begin{cases}
        {\color{black}LM}, & u=0,\ \mu_1=\mu_2,\\
        0,  & 0<|u|<L,\ \mu_1=\mu_2, \\
        0,  & |u|<L, \mu_1\neq \mu_2.
      \end{cases}
    \end{equation}  
  Then $\textbf{C}$ is said to be a $(K,M,L)$-MOGCS. Moreover, $\textbf{C}$ is said to be CCC set if $K=M$ and we denote it as $(K,K,L)$-CCC. Also, each code of $\textbf{C}$ is called a GCS \cite{SaMaLi}.
\end{definition}

For a GCS $\mathbf{C_{\mu}}$, following property holds \cite{Shen2022construction}
\begin{equation}
   \arrowvert p_{\mathbf{a}_{0}^{\mu}}(z)\arrowvert^2+\arrowvert p_{\mathbf{a}_{1}^{\mu}}(z)\arrowvert^2+\hdots+\arrowvert p_{\mathbf{a}_{M-1}^{\mu}}(z)\arrowvert^2=~\text{Constant}. 
\end{equation}
\begin{definition}
  Let $\boldsymbol{\mathscr{Z}}=\{\mathbf{z}_0,\mathbf{z}_1,\hdots,\mathbf{z}_{N-1}\}$ be a set of $N$ sequences of length $L$ each, i.e., $\mathbf{z}_i=(\mathbf{z}_{i0},\mathbf{z}_{i1},\hdots,\mathbf{z}_{iL-1}),\ \ 0\leq i\leq N-.$ Then, $\boldsymbol{\mathscr{Z}}$ is said to be $(N,Z,L)$-ZCZ sequence set if for $0\leq i,j\leq N-1$, ${\boldsymbol{\mathscr{Z}}}$ satisfies following,
  \begin{equation}
      \phi(\mathbf{z}_i,\mathbf{z}_j)(u)=
      \begin{cases}
       0, & i=j\ \text{and}\ 1\leq |u|\leq Z,\\
       0, & i\neq j \ \text{and}\ 0\leq |u|\leq Z,\\
       L, & i=j\ \text{and}\ u=0,
       \end{cases}
  \end{equation}
  where $Z$ is called ZCZ width.
\end{definition} 

\begin{definition}\emph{(Tang-Fan-Matsufuji Bound \cite{TaFaMa})}
Let $\boldsymbol{\mathscr{Z}}$ be any ZCZ sequence set with parameter $(N,Z,L)$. Then, {$N(Z+1)\leq L.$} If $\boldsymbol{\mathscr{Z}}$ achieves the Tang-Fan-Matsufuji Bound with equality then $\boldsymbol{\mathscr{Z}}$ is said to be optimal. Further, if for the larger value of $N$, $N(Z+1)\approx L$, then $\boldsymbol{\mathscr{Z}}$ is said to be asymptotically optimal. However, for the binary case, it is widely accepted that the bound is reduced to $2NZ\leq L$ {\cite{matsu}}. 
   \label{def4}
\end{definition}
\subsection{Restricted and Squeezed Vector/Sequence}
Let $p$ be a prime and $f: \mathbb{Z}_p^n \rightarrow$ $\mathbb{Z}_p$ be an EBF with $n$ variables $x_1, x_2, \cdots, x_n$. Let $J=\left(j_1, j_2, \cdots, j_k\right)$ be a list of $k$ indices with $1 \leq j_1<\cdots<j_k \leq n$ and write $\mathbf{x}_J=$ $\left(x_{j_1}, x_{j_2}, \cdots, x_{j_k}\right)$. Let $\mathbf{c}=\left(c_1, c_2, \cdots, c_k\right)$ be a $p$-ary vector of length $k$. Then we define the restricted complex valued sequence, $\psi(\left.f\right|_{\mathbf{x}_J=\mathbf{c}})$, corresponding to $\left.f\right|_{\mathbf{x}_J=\mathbf{c}}$ with component $i=\sum_{j=1}^m i_j p^{j-1}$ equal to $\omega^{f_i}$ if $i_{j_\alpha}=c_\alpha$ for $1 \leq \alpha \leq k$, and equal to 0 otherwise. The indices of the non-zero entries in the restricted vector is given by
$$
i=\sum_{\alpha=0}^{k-1} c_\alpha p^{j_\alpha}+\sum_{j \notin J} i_j p^j, \quad i_j\in \mathbb{Z}_p.
$$
It is clear that restricted sequences contains zero entries. So, If we remove all the zeros from a restricted sequence then we call remaining sequence as a squeezed sequence. It is clear that squeezed sequence has length $p^{n-k}$.

\subsection{Spectrally-Null-Constrained (SNC) Sequences \cite{shen2023constructions}}
Let $\boldsymbol{a}=\left(a_0, a_1, \cdots, a_{N-1}\right)$ be a complex-valued sequence, the support of $a$ is $\Omega=\left\{x \in \mathbb{Z}_N: a_x \neq 0\right\}$. If $\Omega^c=\left\{x \in \mathbb{Z}_N: x \notin \Omega\right\}$ is not an empty set, $\boldsymbol{a}$ is called a spectrally null constrained (SNC) sequence; otherwise, it is called a traditional sequence. ZCZ sequences such that each sequence is SNC are called SNC-ZCZ sequences. Correspondingly, if there is at least one SNC sequence in the CCC, the CCC is referred to as an SNC-CCC.

\subsection{Generalised Reed-Muller (GRM) Code}
For a prime $p$, consider the polynomial ring $\mathbb{F}_p[x_1,x_2,\hdots,x_l]$ with $l$ variables where $\mathbb{F}_p$ is a field with $p$ elements. GRM codes with parameters $l$ and $\vartheta$ consist of all the sequences corresponding to polynomials with $l$ variables and degree no larger than $\vartheta$.
\begin{definition}
The $\vartheta$-th order $\vartheta\leq p$, $p$-ary GRM code, $GRM_p(l,\vartheta)$, is defined as the following set of $p$-ary vectors
\begin{equation}
GRM_p(l,\vartheta)=\{\xi(f):f\in\mathbb{F}_p[x_1,x_2,\hdots,x_l],~deg(f)< \vartheta\} .
\end{equation}
\end{definition}
Furthermore, it is noteworthy that we always have $x_i^p=x_i$ in $\mathbb{F}_p$, so we only need to consider the polynomials in which the degree of each $x_i$ is no larger than $p-1$. All such polynomials with a degree no larger than $\vartheta$ are linear combinations of the following set of monomials
\begin{equation}
    \mathcal{N}(r)=\Big\{x_1^{j_1}x_2^{j_2},\hdots x_l^{j_l}:j=0,1,2\hdots,p-~\text{and}~\sum_{i=1}^{l}{j_i}\leq \vartheta\Big\}.
\end{equation}
Using combinatorics, we can always prove that the number of monomials in the set is $\mid\mathcal{N}\mid=\sum_{d=0}^{\vartheta}{{l-1+d\choose d}}$. It is well known that $\mathbb{F}_p[x_1,x_2,\hdots,x_l]$ forms a vector space over the field $\mathbb{F}_p$. Further, it can be proved that $GRM_p(l,\vartheta)$ is its subspace. Therefore, $GRM_p(l,\vartheta)$ is a linear code with code length $p^l$ and code dimension $\sum_{d=0}^{\vartheta}{{l-1+d\choose d}}$. Moreover, if we arrange sequences corresponding to monomials in $\mathcal{N}$ as the rows of a matrix, then this matrix forms a generator matrix of $GRM_p(l,\vartheta)$.

\begin{lemma}[\cite{KaSh}]
   For $\lambda=p^n$ and $n\geq 1$, minimum Hamming distance of $GRM_{\lambda}(m,r)$ is $(R+1).p^Q,$ where $R$ is the remainder and $Q$ is the quotient resulting from dividing $m(p-1)-r$ by $p-1$.\label{l0}  
\end{lemma}


\section{Proposed Construction of SNC-CCC}
In this section, we propose a class of EBFs which leads to the constructions of SNC-CCC. 

Let $n\geq 1, k< n, t\leq n-k$. For $J=\{j_1,j_2,\hdots,j_k\} \subset \{1,2,\hdots,n\}$ and $1\leq j_1 <j_2<\hdots<j_k\leq n$, let us assume $\mathbf{x}_J=(x_{j_1},x_{j_2},\hdots,x_{j_k})$. Consider the set $\{1,2,\hdots,n-k\}$ being partitioned into $t$ sets, namely, $V_1,V_2,\hdots,V_t$. Further, let $\pi_r$ be a bijective mapping from $\{1,2,\hdots,m_r\}$ to $V_r$ where $m_r=|V_r|\geq 1,~\forall r=1,2,\hdots,t.$ Define an EBF $f:\mathbb{Z}_p^n\rightarrow \mathbb{Z}_p$ as
\begin{equation}\label{f}
        f\mid_{x_J=\mathbf{c}}=\sum_{r=1}^{t}{\sum_{\alpha=1}^{m_r-1}{x_{\pi_r(\alpha)}x_{\pi_r(\alpha+1)}}}+\sum_{j=1}^{n}{d_{\alpha}x_{\alpha}},
    \end{equation}
where $d_{\alpha}\in \mathbb{Z}_{p}$. Further, let $c\in \mathbb{Z}_{p^k}$ such that $\mathbf{c}=(c_1,c_2,\hdots,c_k)$ is the $p$-ary representation of $c$. Now, for $k=0$, we present a result from \cite{SaLiMa}.
\begin{lemma}
    For each $0\leq h \leq p^t-1$, consider the ordered set given by
    \begin{equation}    \label{palash-CCC}
    C_h=\left\{\mathbf{c}_h^s\equiv\psi\left(f+\sum_{r=1}^{t}{x_{\pi_r(1)}s_{r}}+\sum_{r=1}^{t}{x_{\pi_r (m_r)} h_r}\right) : s_r\in \mathbb{Z}_p\right\},
\end{equation}
where $(h_1,h_2,\hdots,h_t)$ and $(s_1,s_2,\hdots,s_t)$ are vector representation of $h$ and $s$ respectively. Then, $\mathscr{C}=\{C_h:0\leq h\leq p^t-1\}$ is a $(p^t,p^t,p^n)$-CCC.
\end{lemma}
Now we are in the position to state one of the main results of this paper.
\begin{theorem}
     For each $0\leq h \leq p^t-1$, consider the ordered set given by
    \begin{equation}    \label{SNC-CSS}
    C_h=\left\{\mathbf{c}_h^s\equiv\psi\left(f+\sum_{r=1}^{t}{x_{\pi_r(1)}s_{r}}+\sum_{r=1}^{t}{x_{\pi_r (m_r)} h_r}\Bigg\arrowvert_{\mathbf{x}_J=\mathbf{c}}\right) : s_r\in \mathbb{Z}_p\right\},
\end{equation}
where $(h_1,h_2,\hdots,h_t)$ and $(s_1,s_2,\hdots,s_t)$ are vector representation of $h$ and $s$ respectively. Then, $\mathscr{C}=\{C_h:0\leq h\leq p^t-1\}$ is a $(p^t,p^t,p^n)$-SNC-CCC. Further, the support of each sequence in each code $C_h$ is given by $\Omega=\{\sum_{i=1}^k{c_ip^{j_i-1}}+\sum_{\alpha=1,\alpha\notin J}^{n}{b_{\alpha}p^{\alpha-1}}: b_\alpha\in \mathbb{Z}_p\}$.\label{SNC-CCC}
\end{theorem}
\begin{proof}
    Let $I=\{1,2,\hdots,n\}\setminus J\equiv \{i_1,i_2,\hdots,i_{n-k}\}$ such that $1\leq i_1<i_2<\cdots<i_{n-k}$ and $\mathbf{x}_I=\{x_{i_1},x_{i_2},\hdots,x_{i_{n-k}}\}$. Now define a transformation of variables as $x_{i_\alpha}\leftrightarrow y_{\alpha}$, which implies $\mathbf{x}_I=(y_1,y_2,\hdots,y_{n-k})\equiv\mathbf{y}$. Further, define $\mathscr{D}=\{D_h:0\leq h \leq p^t-1\}$, where 
    \begin{equation}\label{compressed-CCC}
    D_h=\left\{g+\sum_{r=1}^{t}{y_{r(1)}s_{r}}+\sum_{r=1}^{t}{y_{r (m_r)} h_r} : s_r\in \mathbb{Z}_p\right\},
\end{equation}
where $g(\mathbf{y})=f\arrowvert_{\mathbf{x}_J=\mathbf{c}}(\mathbf{x}_I),~x_{\pi_r(1)}=y_{r(1)}$ and $x_{\pi_r(m_r)}=y_{r(m_r)}$. Then by \emph{Lemma} \ref{palash-CCC}, $\mathscr{D}$ is a $(p^t,p^t,p^{n-k})$-CCC. Let us define a set of polynomials as 
\begin{equation}\label{SNC-Poly}
    \mathscr{P}_h=\left\{\omega^{d_h(\mathbf{y})} z^{K_0+\epsilon(\mathbf{y})}:d_h(\mathbf{y})\in D_h \right\},
\end{equation}
where $K_0=\sum_{\alpha=1}^k{c_\alpha p^{j_{\alpha}-1}}$ and $\epsilon(\mathbf{y})=\sum_{\alpha=1}^{n-k}{y_\alpha p^{i_{\alpha}-1}}$.
Next, we claim following 
\begin{claim}
    $\mathscr{P}_h$ is the polynomial representation of sequences in $C_h$.
\end{claim}
    To prove the claim, we need to prove following
    \begin{enumerate}
        \item Sequences corresponding to polynomials in $\mathscr{P}_h$ and sequences in $C_h$ have non-zero entries at the same position.
        \item Squeezed sequences corresponding to polynomials in $\mathscr{P}_h$ are same as the  squeezed sequences corresponding to sequences in $C_h$.
    \end{enumerate}
    Using \eqref{SNC-CSS}, \eqref{compressed-CCC}, and \eqref{SNC-Poly} both the properties can be inferred directly. Hence $\mathscr{P}_h$ can alternatively be written as
    \begin{equation}
        \mathscr{P}_h=\{p_{\mathbf{c}_h^s}(z):\mathbf{c}_h^s\in C_h\}.
    \end{equation}
    Now, we prove that $\mathscr{P}_h$ is a GCS and a similar calculation can be done to prove $\mathscr{P}_h$ and $\mathscr{P}_{h'}$ are mutually orthogonal for $0\leq h\neq h'\leq p^t-1$. Consider
    \begin{equation}\label{GCS-proof}
    \begin{aligned}
        \sum_{s=0}^{p^t-1}{p_{\mathbf{c}_h^s}(z)p_{(\mathbf{c}_h^s)^*}(-z)}=&\sum_{d_h\in D_h}{\sum_{y,y'=0}^{p^{n-k}-1}{\omega^{d_h(\mathbf{y})-d_h(\mathbf{y}')} z^{\epsilon(\mathbf{y})-\epsilon(\mathbf{y}')}}}\\
        =& \mathscr{L}_1+\mathscr{L}_2,
    \end{aligned}        
    \end{equation}
    where
    \begin{equation}
        \mathscr{L}_1=\sum_{d_h\in D_h}{\sum_{y=y'=0}^{p^{n-k}-1}{\omega^{d_h(\mathbf{y})-d_h(\mathbf{y}')} z^{\epsilon(\mathbf{y})-\epsilon(\mathbf{y}')}}}
    \end{equation}
    and
    \begin{equation}
        \mathscr{L}_2=
        \sum_{d_h\in D_h}{\sum_{y\neq y'}{\omega^{d_h(\mathbf{y})-d_h(\mathbf{y}')} z^{\epsilon(\mathbf{y})-\epsilon(\mathbf{y}')}}}.
    \end{equation}
For $y=y'$, the value of $\omega^{d_h(\mathbf{y})-d_h(\mathbf{y}')} z^{\epsilon(\mathbf{y})-\epsilon(\mathbf{y}')}=1$. Hence, $\mathscr{L}_1=$constant. Next, let $\beta_r=\sum_{d_h\in D_h}{\sum_{y-y'=r}{\omega^{d_h(\mathbf{y})-d_h(\mathbf{y}')}}}$, $1-p^{n-k}\leq r \leq p^{n-k}-1$. Then $\mathscr{L}_2$ can be rewritten as 
\begin{equation}
   \mathscr{L}_2=\sum_{r\neq 0}{\beta_r z^{\epsilon(\mathbf{y})-\epsilon(\mathbf{y}')}}. 
\end{equation}
Since, $d_h\in D_h$, therefore by the property of GCS, $\beta_r=0$, $\forall r$. Also, it can be noted that when $y-y'=r$, then $\mathbf{y}-\mathbf{y}'$ is some fixed vector and hence,
$\epsilon(\mathbf{y}-\mathbf{y}')=\sum_{\alpha=1}^{n-k}{(y_\alpha-y'_\alpha) p^{i_{\alpha}-1}}$ is a constant which implies that $\mathscr{L}_2=0$. Therefore, from \eqref{GCS-proof}
\begin{equation}
    \sum_{s=0}^{p^t-1}{p_{\mathbf{c}_h^s}(z)p_{(\mathbf{c}_h^s)^*}(-z)}=\text{constant}.
\end{equation}
Thus, $\mathscr{P}_h$ is a GCS. 
\end{proof}
\begin{example}\label{SNC-CCCex}
    For $n=6$, $J=\{3\}$, let $\{1,2,4,5,6\}$ is partitioned into $2$ sets, i.e., $V_1=\{1,2\}$ and $V_2=\{4,5,6\}$ with $\pi_1(1)=1,~\pi_1(2)=2,~\pi_2(1)=5,~\pi_2(2)=4,$ and $\pi_2(3)=6$. Let $f:\mathbb{Z}_2^6\rightarrow \mathbb{Z}_2$ be defined as    $f=x_1x_2+x_2x_3+x_3x_4+x_4x_5+x_4x_6$, then $f\arrowvert_{x_3=1}=x_1x_2+x_4x_5+x_4x_6+x_2+x_4$. For $0\leq h\leq 3$, define 
    \begin{equation}
        C_h=\left\{(f+x_1s_1+x_5s_2+x_2h_1+x_6h_2)\arrowvert_{x_3=1}:s_1,s_2\in \{0,1\}\right\}.
    \end{equation}
Then $\mathscr{C}=\{C_0,C_1,C_2,C_3\}$ is a $(4,4,64)$-SNC-CCC as shown in \emph{Table} \ref{tableSNC-CCC}. 
\end{example}
\begin{table}[ht]
\caption{SNC-CCC corresponding to \emph{Example} 
\label{tableSNC-CCC}
\ref{SNC-CCCex}}
\begin{tabular}{|c|l|}
    \hline
       \multirow{4}{*}{$C_0$} &000011-110000-1-11-1000011-11000011-11000011-11000011-11000011-110000-1-11-1\\
       &00001-1-1-10000-111100001-1-1-100001-1-1-100001-1-1-100001-1-1-100001-1-1-10000-1111\\
       &000011-110000-1-11-10000-1-11-10000-1-11-1000011-11000011-110000-1-11-1000011-11\\
       &00001-1-1-10000-11110000-11110000-111100001-1-1-100001-1-1-10000-111100001-1-1-1\\ \hline
       \multirow{4}{*}{$C_1$} & 0000111-10000-1-1-110000111-10000111-10000111-10000111-10000111-10000-1-1-11\\
       & 00001-1110000-11-1-100001-11100001-11100001-11100001-11100001-1110000-11-1-1\\
       & 0000111-10000-1-1-110000-1-1-110000-1-1-110000111-10000111-10000-1-1-110000111-1\\
       & 00001-1110000-11-1-10000-11-1-10000-11-1-100001-11100001-1110000-11-1-100001-111\\ \hline
        \multirow{4}{*}{$C_2$} & 000011-110000-1-11-1000011-11000011-110000-1-11-10000-1-11-10000-1-11-1000011-11\\
        & 00001-1-1-10000-111100001-1-1-100001-1-1-10000-11110000-11110000-111100001-1-1-1\\
        & 000011-110000-1-11-10000-1-11-10000-1-11-10000-1-11-10000-1-11-1000011-110000-1-11-1\\
        & 00001-1-1-10000-11110000-11110000-11110000-11110000-111100001-1-1-10000-1111\\ \hline
        \multirow{4}{*}{$C_3$} & 0000111-10000-1-1-110000111-10000111-10000-1-1-110000-1-1-110000-1-1-110000111-1\\
        & 00001-1110000-11-1-100001-11100001-1110000-11-1-10000-11-1-10000-11-1-100001-111\\
        & 0000111-10000-1-1-110000-1-1-110000-1-1-110000-1-1-110000-1-1-110000111-10000-1-1-11\\
        & 00001-1110000-11-1-10000-11-1-10000-11-1-10000-11-1-10000-11-1-100001-1110000-11-1-1\\ \hline
    \end{tabular}
    \end{table} 

\section{Proposed construction of SNC-ZCZ Sequences}
In the next theorem, a construction of SNC-ZCZ sequence sets is presented. Before presenting the theorem, let us define an EBF $g(x_{n+1},x_{n+2},\hdots,x_{n+t+1})$ on $t+1$ variables as
    \begin{equation}
    \begin{split}      g(x_{n+1},x_{n+2},\hdots,x_{n+t+1})=\sum_{r=2}^{t+1}{c_rx_{n+r}x_{n+1}}+\hspace{-.3cm}\sum_{2\leq \mu<\nu\leq t}\hspace{-.3cm}{d_{\mu \nu}x_{n+\mu} x_{n+\nu}}+\sum_{\beta=1}^{t+1}{e_\beta x_{n+\beta}}+e',\label{eq:6}
    \end{split}
  \end{equation}
  $\text{where} \ c_r,d_{\mu \nu}, e_{\beta},e'\in \mathbb{Z}_p,\ {\color{black}c_{t+1}\neq 0}\ $.
  \begin{remark}
      From \eqref{eq7}, it can be verified that $\mathscr{C}$ is also a CCC with respect to periodic correlation. This idea has been utilized in the proof of \emph{Theorem} \ref{SNC-ZCZ}.
  \end{remark}
\begin{theorem}\label{SNC-ZCZ}
For $0\leq c\leq p^k-1$, define a set $\boldsymbol{\mathscr{Z}}=\{\mathbf{z}_h:\ h=0,1,\hdots,p^t-1\}$ such that
\begin{equation}    \mathbf{z}_h=\psi\left(\left(f+g+\sum_{r=1}^{t}{x_{\pi_r(1)}x_{n+r}}+\sum_{r=1}^{t}{x_{\pi_r (m_r)} h_r}\right)\Bigg\arrowvert_{\mathbf{x}_J=\mathbf{c}}\right),
\end{equation}
where $h=(h_1,h_2,\hdots,h_t)\cdot (1,p,\hdots,p^{t-1})$ and $c=\mathbf{c}\cdot (1,p,\hdots,p^{k-1})$. Then, the set $\boldsymbol{\mathscr{Z}}$ is a $(p^t,(p-1)p^{n},p^{n+t+1})$-SNC-ZCZ. Further, the support of each sequence $\mathbf{z}_h$ is given by $\Omega=\{\sum_{i=1}^k{c_ip^{j_i-1}}+\sum_{\alpha=1,\alpha\notin J}^{n+t+1}{b_{\alpha}p^{\alpha-1}}: b_\alpha\in \mathbb{Z}_p\}$.
\end{theorem}
\begin{proof}
Let $\mathbf{x}_T=[x_{n+1},x_{n+2},\hdots,x_{n+t+1}]$, $\mathbf{l}=[l_{n+1},l_{n+2},\hdots,l_{n+t+1}],\mathbf{l}'=[l_{n+1}',l_{n+2}',$ $\hdots,l_{n+t+1}']$, $\mathbf{l}_1=[l_{n+1},l_{n+2},\hdots,l_{n+t}]$, $\mathbf{l}'_1=[l'_{n+1},l'_{n+2},\hdots,l'_{n+t}]$. For $0\leq h,h'\leq p^t-1$, the PCCF of $\mathbf{z}_h$ and $\mathbf{z}_{h^{\prime}}$ is given by 
\begin{equation}
\begin{aligned}
 P&= \phi\left(\mathbf{z}_h, \mathbf{z}_{h^{\prime}}\right)(u) \\
&=\phi(f+g+\sum_{r=1}^{t} x_{\pi_r(1)} x_{n+r}+\left.\sum_{r=1}^{t} x_{\pi_r\left(m_r\right)} h_r\right|_{\mathbf{x}_J=\mathbf{c}},\\
&\hspace*{4.5cm}\left.f+g+\sum_{r=1}^{t} x_{\pi_{r(1)}} x_{n+r}+\left.\sum_{r=1}^{t} x_{\pi_r\left(m_r\right)} h_r^{\prime}\right|_{\mathbf{x}_J=\mathbf{c}}\right)(u)\\
&=\sum_{\mathbf{l},\mathbf{l}'}{\omega_p^\delta\phi\left(f+\sum_{r=1}^{t} x_{\pi_r(1)} l_{n+r}+\left.\sum_{r=1}^{t} x_{\pi_r\left(m_r\right)} h_r\right|_{\mathbf{x}_J\mathbf{x}_T=\mathbf{c}\mathbf{l}},\right.}\\
&\hspace*{4.5cm}\left.f+\sum_{r=1}^{t} x_{\pi_{r(1)}} l'_{n+r}+\left.\sum_{r=1}^{t} x_{\pi_r\left(m_r\right)} h_r^{\prime}\right|_{{\mathbf{x}_J\mathbf{x}_T=\mathbf{c}\mathbf{l}'}}\right)(u),
\end{aligned}
\end{equation}
where,
\begin{equation}      
\delta=\sum_{r=2}^{t+1}c_r(l_{n+r}l_{n+1}-l_{n+r}'l_{n+1}')+\hspace{-.3cm}\sum_{2\leq \mu<\nu\leq t}\hspace{-.3cm}{d_{\mu\nu}(l_{n+\mu}l_{n+\nu}-l_{n+\mu}'l_{n+\nu}')}+\sum_{\beta=1}^{t+1}{e_\beta(l_{n+\beta}-l_{n+\beta}')}.
  \end{equation}
  $P$ can further be simplified to,
\begin{equation}
\begin{aligned}
P=&\sum_{\mathbf{l}_1,\mathbf{l}'_1}\omega_p^{\delta'}\sum_{l_{n+t+1},l_{n+t+1}'}\omega_p^{c_{t+1}(l_{n+t+1}l_{n+1}-l_{n+t+1}'l_{n+1}')+e_{t+1}(l_{n+t+1}-l_{n+t+1}')}\\
&\times\phi\left(f+\sum_{r=1}^{t} x_{\pi_r(1)} l_{n+r}+\left.\sum_{r=1}^{t} x_{\pi_r\left(m_r\right)} h_r\right|_{\mathbf{x}_J\mathbf{x}_T=\mathbf{c}\mathbf{l}},\right.\\
&\hspace*{4.5cm}\left.f+\sum_{r=1}^{t} x_{\pi_{r(1)}} l'_{n+r}+\left.\sum_{r=1}^{t} x_{\pi_r\left(m_r\right)} h_r^{\prime}\right|_{{\mathbf{x}_J\mathbf{x}_T=\mathbf{c}\mathbf{l}'}}\right)(u),
\end{aligned}\label{eq:31}
\end{equation}
where,
\begin{equation}      
\delta'=\sum_{r=2}^{t}c_r(l_{n+r}l_{n+1}-l_{n+r}'l_{n+1}')+\hspace{-.3cm}\sum_{2\leq \mu<\nu\leq t}\hspace{-.3cm}{d_{\mu\nu}(l_{n+\mu}l_{n+\nu}-l_{n+\mu}'l_{n+\nu}')}+\sum_{\beta=1}^{t}{e_\beta(l_{n+\beta}-l_{n+\beta}')}.
  \end{equation}
  Now we have two cases 
  \begin{case}[$l_{n+t+1}=l_{n+t+1}'$]
       Let $l_{n+1}=l'_{n+1}$. Since, $0\leq u\leq (p-1)p^n$, therefore, we have $l_{n+r}=l'_{n+r}\ \forall r=2,3,\hdots,t$. Hence, \eqref{eq:31} will reduce to 
       \begin{equation}
       \begin{aligned}
       P&=\sum_{l_{n+t+1}}\sum_{\mathbf{l}_1}\phi\left(f+\sum_{r=1}^{t} x_{\pi_r(1)} l_{n+r}+\sum_{r=1}^{t} x_{\pi_r(m_r)} h_r\right|_{\mathbf{x}_J\mathbf{x}_T=\mathbf{c}\mathbf{l}},\\
       &\hspace*{4cm}\left.\left.f+\sum_{r=1}^{t} x_{\pi_{r(1)}} l_{n+r}+\sum_{r=1}^{t} x_{\pi_r\left(m_r\right)} h_r^{\prime}\right|_{\mathbf{x}_J\mathbf{x}_T=\mathbf{c}\mathbf{l}}\right)(u).
       \end{aligned}
       \end{equation}
       From \emph{Theorem} \ref{SNC-CCC} the value of 
       \begin{equation}
           \begin{aligned}               &\sum_{\mathbf{l}_1}\phi\left(f+\sum_{r=1}^{t} x_{\pi_r(1)} l_{n+r}+\sum_{r=1}^{t} x_{\pi_r(m_r)} h_r\right|_{\mathbf{x}_J\mathbf{x}_T=\mathbf{c}\mathbf{l}},\\
       &\hspace*{4cm}\left.\left.f+\sum_{r=1}^{t} x_{\pi_{r(1)}} l_{n+r}+\sum_{r=1}^{t} x_{\pi_r\left(m_r\right)} h_r^{\prime}\right|_{\mathbf{x}_J\mathbf{x}_T=\mathbf{c}\mathbf{l}}\right)(u)
           \end{aligned}
       \end{equation}
       is zero for all non-zero time shifts. Further, if $l_{n+1}\neq l_{n+1}'$ then \eqref{eq:31} will reduce to
       \begin{equation}
\begin{aligned}
P&=\sum_{\mathbf{l}_1,\mathbf{l}'_1}\omega_p^{\delta'}\sum_{l_{n+t+1}}\hspace*{-.2cm}\omega_p^{c_{t+1}l_{n+t+1}(l_{n+1}-l_{n+1}')}\phi\left(f+\sum_{r=1}^{t} x_{\pi_r(1)} l_{n+r}+\sum_{r=1}^{t} x_{\pi_r(m_r)} h_r\right|_{\mathbf{x}_J\mathbf{x}_T=\mathbf{c}\mathbf{l}},\\
&\hspace*{5cm}\left.\left.f+\sum_{r=1}^{t} x_{\pi_{r(1)}} l_{n+r}+\sum_{r=1}^{t} x_{\pi_r(m_r)} h_r^{\prime}\right|_{\mathbf{x}_J\mathbf{x}_T=\mathbf{c}\mathbf{l}'}\right)(u)\\
&=\sum_{\mathbf{l}_1,\mathbf{l}'_1}\omega_p^{\delta'}\phi\left(f+\sum_{r=1}^{t} x_{\pi_r(1)} l_{n+r}+\sum_{r=1}^{t} x_{\pi_r(m_r)} h_r\right|_{\mathbf{x}_J\mathbf{x}_T=\mathbf{c}\mathbf{l}},\\
&\hspace*{1cm}\left.\left.f+\sum_{r=1}^{t} x_{\pi_{r(1)}} l_{n+r}+\sum_{r=1}^{t} x_{\pi_r(m_r)} h_r^{\prime}\right|_{\mathbf{x}_J\mathbf{x}_T=\mathbf{c}\mathbf{l}'}\right)(u)\sum_{l_{n+t+1}}\omega_p^{c_{t+1}l_{n+t+1}(l_{n+1}-l_{n+1}')}.
\end{aligned}
\end{equation}
Since $\sum_{l_{n+t+1}}\omega_p^{c_{t+1}l_{n+t+1}(l_{n+1}-l_{n+1}')}=0$. Therefore, $P=0$.
       
  \end{case}

  \begin{case}[$l_{n+t+1}\neq l_{n+t+1}'$] For $0\leq u\leq (p-1)p^n$, $l_{n+t+1}-l_{n+t+1}'=-1\mod p$ and $l_{n+1}-l_{n+1}'=-1\mod p$. This implies that $l_{n+t+1}l_{n+t+1}'-l_{n+1}l_{n+1}'=-1-l_{n+t+1}-l_{n+1}$. Now \eqref{eq:31} can be written as
  \begin{equation}
 \begin{aligned} P&=\sum_{\mathbf{l}_1,\mathbf{l}'_1}\omega_p^{\delta'}\sum_{l_{n+t+1}}\omega_p^{c_{t+1}(-1-l_{n+t+1}-l_{n+1})-e_{t+1}}\phi\left(f+\sum_{r=1}^{t} x_{\pi_r(1)} l_{n+r}\right.\\
 &\hspace*{1cm}\left.\left.\left.+\sum_{r=1}^{t} x_{\pi_r(m_r)} h_r\right|_{\mathbf{x}_J\mathbf{x}_T=\mathbf{c}\mathbf{l}},f+\sum_{r=1}^{t} x_{\pi_{r(1)}} l_{n+r}+\sum_{r=1}^{t} x_{\pi_r(m_r)} h_r^{\prime}\right|_{\mathbf{x}_J\mathbf{x}_T=\mathbf{c}\mathbf{l}'}\right)(u)\\
&=\sum_{\mathbf{l}_1,\mathbf{l}'_1}\omega_p^{\delta'-c_{t+1}(1+l_{n+1})-e_{t+1}}\phi\left(f+\sum_{r=1}^{t} x_{\pi_r(1)} l_{n+r}+\sum_{r=1}^{t} x_{\pi_r(m_r)} h_r\right|_{\mathbf{x}_J\mathbf{x}_T=\mathbf{c}\mathbf{l}},\\
&\hspace*{2cm}\left.\left.f+\sum_{r=1}^{t} x_{\pi_{r(1)}} l_{n+r}+\sum_{r=1}^{t} x_{\pi_r(m_r)} h_r^{\prime}\right|_{\mathbf{x}_J\mathbf{x}_T=\mathbf{c}\mathbf{l}'}\right)(u)\sum_{l_{n+t+1}}\omega_p^{-c_{t+1}l_{n+t+1}}.
\end{aligned}
\end{equation}
Since $\sum_{l_{n+t+1}}\omega_p^{-c_{t+1}l_{n+t+1}}=0$, therefore, $P=0$.
  
  \end{case}
\end{proof}

\begin{example}\label{ex1}
    Let $f:\mathbb{Z}_2^3\rightarrow\mathbb{Z}_2$ be an EBF defined as $f(x_1,x_2,x_3)=x_1x_2+x_1x_3+x_1+x_3$. Further, let $J=\{1\}$ such that $V_1=\{x_2\}$ and $V_2=\{x_3\}$. Now, define another EBF $g:\mathbb{Z}_2^3\rightarrow\mathbb{Z}_2$ as $g(x_4,x_5,x_6)=x_4x_5+x_4x_6$. For, $c=1$
    \begin{equation}
        \boldsymbol{\mathscr{Z}}=
        \begin{bmatrix}
            \mathbf{z}_0\\
            \mathbf{z}_1\\
            \mathbf{z}_2\\
            \mathbf{z}_3
        \end{bmatrix}
        =
        \begin{bmatrix}
            f+g+x_2x_4+x_3x_5+x_2\cdot0 +x_3\cdot 0|_{x_1=1}\\
            f+g+x_2x_4+x_3x_5+x_2\cdot1 +x_3\cdot 0|_{x_1=1}\\
            f+g+x_2x_4+x_3x_5+x_2\cdot0 +x_3\cdot 1|_{x_1=1}\\
            f+g+x_2x_4+x_3x_5+x_2\cdot1 +x_3\cdot 1|_{x_1=1}            
        \end{bmatrix}.
    \end{equation}  
\begin{table}[ht]
\caption{SNC-ZCZ sequences corresponding to \emph{Example} \ref{ex1}}
\label{tableSNC-ZCZ}
    \begin{tabular}{|c|l|}
    \hline
       \multirow{4}{*}{$\boldsymbol{\mathscr{Z}}$} &0-1010-1010-10-10-10-10-101010-10-10-101010-1010-10101010101010-10-1010-10-10101  \\ 
        & 0-10-10-10-10-1010-1010-10-101010-101010-10-10-10-10101010101010-10-1010-10-10101\\ 
        & 0-101010-10-10-101010-1010-1010-10-10-10-10-101010-101010-10-1010-1010-10-10-10-10-1\\ 
        &0-10-101010-101010-10-10-10-10-10-1010-1010-10-10101010-10-101010101010-1010-101\\ \hline
    \end{tabular}
    \end{table}    
    Then, $\boldsymbol{\mathscr{Z}}$ is $(4,8,64)$-SNC-ZCZ sequence set as shown in \emph{Table} \ref{tableSNC-ZCZ}
    
\end{example}

We have proposed the construction of SNC-ZCZ sequences. Next, corollary demonstrate the relation between the proposed SNC-ZCZ sequences and traditional ZCZ sequences.
\begin{corollary}
    In \emph{Theorem} \ref{SNC-ZCZ}, let $J$ be an empty set. Then the set $\boldsymbol{\mathscr{Z}}$ is a traditional $(p^t,(p-1)p^n,p^{n+t+1})$-ZCZ sequence set, i.e., $\mid\Omega\mid=p^{n+t+1}$.\label{trad-ZCZ}
\end{corollary}
\begin{remark}
    It can be verified that for the larger value of $p$, $N(Z+1)\approx L$. Hence, proposed ZCZ sequences are asymptotically optimal.
\end{remark}

\begin{remark}
    For the fixed $J$ and different $\mathbf{c}$, the support of corresponding SNC-ZCZ sequences is different but have same size. On the other hand, for the different $J$'s, the corresponding SNC-ZCZ sequences have different support set as well as different size. This shows that the proposed SNC-ZCZ sequences have flexible support.
\end{remark}

\section{Relation Between Proposed traditional ZCZ Sequences and GRM Codes}\label{sec4}

This section establishes a relation between constructed ZCZ sequences and GRM codes. For ease of notation, we assume $n+t+1=n'$ and $\mathbf{z}_h$ refers to function instead of sequence as shown in \eqref{zih_2}. Using \emph{Corollary} \ref{trad-ZCZ} we can write 
\begin{equation}\label{zih_2}
    \boldsymbol{\mathscr{Z}}=\left\{\mathbf{z}_{h}=\mathcal{Q}+\sum_{j=1}^{n}{d_{\alpha}x_{\alpha}}
    +\sum_{\beta=1}^{t+1}{e_\beta x_{n+\beta}}+\sum_{r=1}^{t}{x_{\pi_r (m_r)} h_r}+e': h_r \in \mathbb{Z}_p\right\},
\end{equation}
where 
\begin{equation}
      \mathcal{Q}=\underbrace{\sum_{r=1}^{t}{\sum_{j=1}^{m_r-1}{x_{\pi_r(j)}x_{\pi_r(j+1)}}}}_{\mathcal{Q}_1}+\underbrace{\sum_{r=2}^{t+1}{c_r x_{n+r}x_{n+1}}}_{\mathcal{Q}_2}+\hspace{-.15cm}\underbrace{\sum_{2\leq \mu<\nu\leq t}\hspace{-.3cm}{d_{\mu \nu}x_{n+\mu} x_{n+\nu}}}_{\mathcal{Q}_3}\\
      +\underbrace{\sum_{r=1}^{t}{x_{\pi_r(1)}x_{n+r}}}_{\mathcal{Q}_4}.\label{eq:60}
\end{equation}

In the following corollary, we see that ZCZ sequence sets with fixed quadratic form $\mathcal{Q}$, partition the second order coset of the first order GRM code, i.e., $\mathcal{Q} + GRM_p(n',1)$

\begin{corollary}\label{corr1}
Let $\mathcal{Q} + GRM_p(n',1)$ be second order coset of the first order GRM code. Then for a fixed quadratic form $\mathcal{Q}$, $\mathcal{Q} + GRM_p(n',1)$ contains $p^{n'+1-t}$ distinct ZCZ sequence sets. 
\end{corollary}
\begin{proof}
  From equation \eqref{zih_2}, it is clear that for a fixed quadratic form $\mathcal{Q}$, the parameters which contribute to different ZCZ sequence sets are $d_\alpha,~e_\beta,$ and $e'$.  But, whenever $x_{\pi_r(m_r)}=x_\alpha$, then respective $d_\alpha$ will no longer contribute to generate different ZCZ sequence sets because $\sum_{r=1}^{t}x_{\pi_r(m_r)}h_r$ contributes to different sequences in a particular ZCZ sequence sets. Therefore, total variables which contributes to different ZCZ sequence sets are $n'+1-t$. Hence total number of different ZCZ sequence sets in $\mathcal{Q} + GRM_p(n',1)$ equals to $p^{n'+1-t}$.
\end{proof}

It can be obtained that if we have a sequence in $\boldsymbol{\mathscr{Z}}$, then it is a codeword belonging to second-order coset $\mathcal{Q} + GRM_p(n',1)$ where $\mathcal{Q}$ is given by \eqref{eq:60}. Moreover, 
if any codeword belongs to the coset $\mathcal{Q} + GRM_p(n',1)$ then it belongs to a certain $(p^t,(p-1)p^n,p^{n+t+1})$-ZCZ sequence set. Hence ZCZ sequence sets with fixed quadratic form $\mathcal{Q}$, partition $\mathcal{Q} + GRM_p(n',1)$.

 In the following corollary, we explicitly determine the number of second-order cosets representatives having the form $\mathcal{Q}+GRM_p(n',1)$. 
 \begin{corollary}\label{corr2}
     The number of coset representative $\mathcal{Q}$ in the second-order coset of the form $\mathcal{Q}+GRM_p(n',1)$ where $\mathcal{Q}$ is as defined in \eqref{eq:60} is given by 
\begin{equation}
    (p-1)p^{\frac{(t-1)t}{2}}\bigg(\sum_{m_1+m_2+\cdots+m_t=n}{\prod_{r=1}^{t}{\frac{m_{r}!}{2}}}\bigg).
\end{equation}
\end{corollary}
\begin{proof}
One can easily see that the two permutations $\pi_r$ and $\pi'_r$, where $\pi'_r(k)=\pi_r(m_r+k-1)$ will generate the same quadratic form. Therefore, we have $\frac{m_{r}!}{2}$ different quadratic forms of the type  
\begin{equation}
  x_{\pi_{r}(1)}x_{\pi_{r}(2)} + x_{\pi_{r}(2)}x_{\pi_{r}(3)} + \cdots+ x_{\pi_{r}(m_r-1)}x_{\pi_{r}(m_{r})},~r=1,2,\hdots t,
\end{equation}
where notations are the same as used in \emph{Lemma} \ref{palash-CCC}.
In addition, we have the condition $m_1 + m_2 + \cdots + m_t = n$
where $1\leq m_r\leq n$. Therefore, we can explicitly determine
    $\sum_{m_1+m_2+\cdots+m_t=n}{\prod_{r=1}^{t}{\frac{m_{r}!}{2}}}$
number of distinct coset representatives of the type ${\mathcal{Q}_1}$. It is to be noted that once we fix ${\mathcal{Q}_1}$ then ${\mathcal{Q}_2},{\mathcal{Q}_3},~\text{and}~ {\mathcal{Q}_4}$ are also fixed except their coefficients, i.e., $c_r$ and $d_{\mu\nu}$ where $c_r,d_{\mu\nu}\in \mathbb{Z}_p$, $c_{t+1}\neq0$. Since ${\mathcal{Q}_2}$ and ${\mathcal{Q}_3}$ contain $t$ and $\frac{(t-2)(t-1)}{2}$ terms respectively and $c_{t+1}$ can take only $p-1$ values. Therefore, We have $(p-1)p^{t-1}p^{\frac{(t-2)(t-1)}{2}}=(p-1)p^{\frac{(t-1)t}{2}}$ number of distinct coset representative of the type ${(\mathcal{Q}_2}+{\mathcal{Q}_3}+{\mathcal{Q}_4})+GRM_p(n',1)$ for the fixed value of ${\mathcal{Q}_1}$. Hence, combining all the discussion, we can conclude that the total number of coset representative of the form $\mathcal{Q}+GRM_p(n',1)$ is 
\begin{equation*}
    (p-1)p^{\frac{(t-1)t}{2}}\bigg(\sum_{m_1+m_2+\cdots+m_t=n}{\prod_{r=1}^{t}{\frac{m_{r}!}{2}}}\bigg).
\end{equation*}
\end{proof}

In the next corollary, it has been concluded that the ZCZ sequence sets lying inside the different cosets of $GRM_p(n',1)$ are different.

\begin{corollary}
   Let $\boldsymbol{\mathscr{Z}}$ and $\boldsymbol{\mathscr{Z}'}$ be two ZCZ sequence sets obtained from \emph{Corollary} \ref{trad-ZCZ}, contained in $\mathcal{Q}+GRM_p(n',1)$ and $\mathcal{Q}'+GRM_p(n',1)$, respectively, where $\mathcal{Q}$ and $\mathcal{Q}'$ are different quadratic forms. Then, $\boldsymbol{\mathscr{Z}}$ and $\boldsymbol{\mathscr{Z}'}$ are different ZCZ sequence sets.   
\end{corollary}
\begin{proof}
   It Suffices to prove that the cosets $\mathcal{Q}+GRM_p(n',1)$ and $\mathcal{Q}'+GRM_p(n',1)$ are either disjoint or same. Let us assume that there exist an element $c$ such that $c\in \mathcal{Q}+GRM_p(n',1) \cap \mathcal{Q}'+GRM_p(n',1)$. Then
   $c=\mathcal{Q}+c_1=\mathcal{Q}'+c_2,$ for some $c_1,c_2\in GRM_p(n',1).$ Which implies $\mathcal{Q}-\mathcal{Q}'=c_2-c_1$. Since, $c_1,c_2\in GRM_p(n',1)$ and $GRM_p(n',1)$ is linear code therefore, $c_2-c_1\in GRM_p(n',1)$, which further implies $\mathcal{Q}-\mathcal{Q}'\in GRM_p(n',1)$. Now for some $w\in GRM_p(n',1)$ let $\mathcal{Q}-\mathcal{Q}'=w$. Then,
   \[\mathcal{Q}+GRM_p(n',1)=\mathcal{Q}'+w+GRM_p(n',1)=\mathcal{Q}'+GRM_p(n',1).\]
   Which follows the result.
\end{proof}

\begin{remark}
  From \emph{Corollary} \ref{corr1}, each cosets of $GRM_p(n',1)$ comprises  $p^{n'+1-t}$ ZCZ sequence sets. Further, from \emph{Corollary} \ref{corr2}, the number of distinct cosets are  $(p-1)p^{\frac{(t-1)t}{2}}\bigg(\sum_{m_1+m_2+\cdots+m_t=n}{\prod_{r=1}^{t}{\frac{m_{r}!}{2}}}\bigg)$. Hence, on combining we obtain $(p-1)p^{\frac{(t-1)t}{2}}\bigg(\sum_{m_1+m_2+\cdots+m_t=n}{\prod_{r=1}^{t}{\frac{m_{r}!}{2}}}\bigg)p^{n'+1-t}$ ZCZ sequence sets. 
\end{remark}
 
\begin{example}
 Let $p=3,n=4,t=2$, $f=x_1x_2+x_3x_4$, and $g=x_5x_6+x_5x_7$. Further, let $\pi_1(1)=1, \pi_1(2)=2,\pi_2(1)=3$, and $\pi_2(2)=4$.
 From \eqref{zih_2},
 \begin{multline}
    \boldsymbol{\mathscr{Z}}=\left\{\mathcal{Q}+d_1x_1+d_3x_3+e_1x_5+e_2x_6+e_3x_7+(d_2+h_1)x_2+(d_4+h_2)x_4\right.\\
    \left.+e':h_1,h_2\in \mathbb{Z}_3\right\}. 
 \end{multline}
 from the above equation, it is clear that for the different values of $d_1,d_3,e_1,e_2,e_3$, and $e'$ the ZCZ sequence set $\boldsymbol{\mathscr{Z}}$ is different. Hence, the coset $\mathcal{Q}+GRM_3(7,1)$ with $\mathcal{Q}=f+g+x_1x_5+x_2x_6$ contains $3^6$ distinct $(3^4,2\cdot3^2,3^7)$-ZCZ sequence sets.
\end{example}

\begin{corollary}
    The minimum Hamming distance of a proposed ZCZ sequence set is at least $(p-1)p^{n'-1}$.
\end{corollary}
\begin{proof}   
Using \emph{Lemma} \ref{l0}, it can be calculated that the minimum Hamming distance of $GRM_p(n',1)$ is $(p-1)p^{n'-1}$. From \emph{Corollary} \ref{corr1}, each ZCZ sequence set $\boldsymbol{\mathscr{Z}}$ is a subset of a second order coset of the first order GRM code, i.e., $\mathcal{Q} + GRM_p(n',1)$. Now, Since the coset of a code has the same minimum Hamming distance as that of code. Therefore, minimum Hamming distance of $\mathcal{Q} + GRM_p(n',1)$ is $(p-1)p^{n'-1}$. Hence, minimum Hamming distance of $\boldsymbol{\mathscr{Z}}$ is at least $(p-1)p^{n'-1}$.  
\end{proof}

\section{Comparison with Existing Works}\label{sec5}

In this section, we compare our proposed construction of traditional/SNC-ZCZ sequences with some existing works. A detailed comparison of traditional ZCZ sequences can be seen in Table \ref{table2}.

\subsection{Comparison with Traditional ZCZ Sequences in \cite{ChWu} and \cite{LiGuPa}}
In \cite{ChWu}, Chen \emph{et al.} proposed ZCZ sequence sets having parameters $(2^k,2^{m-k-1},2^m)$, where $k\ \text{and}\ m$ are positive integers such that 
$m\geq 2,$ and $k\leq m-1$. In our proposed construction, if we take $p=2,~t=k,$ and $n=m-k-1$, then the parameters of our construction meet the parameters in \cite{ChWu}. Moreover, in \cite{LiGuPa}, authors provided ZCZ sequence sets having parameters $(2^{k+1},2^n,2^{n+k+2})$. Construction in \cite{LiGuPa} was based on GBFs. Again, these parameters also appear as a special case of our parameters, i.e., if we take $p=2$ and $t=k+1$ in our proposed construction, then we get the parameters of \cite{LiGuPa}. Hence, our construction covers more lengths and can be seen as a generalization of these constructions.

\begin{sidewaystable}
\caption{Comparison with \cite{TaFa},\cite{ChWu},\cite{Th}, \cite{LiGuPa}, \cite{ZhDz}, \cite{ZhPa}, and \cite{TaYu}}
    \begin{tabular}{llllll}
    \hline
       \textbf{Reference}  & \textbf{Based on} & \textbf{Parameter ($N$,$Z,L$)} & \textbf{Direct/Indirect} &
       \textbf{Optimality}\\ \hline
        \cite{TaFa} & MOGCS & $(2^n,Z_{cz}+1,2^{n+1}Z_{cz}),\ n \geq 1$ & Indirect & Optimal for binary case\\
        \hline
        
        \cite{ChWu} &  GBF &  {\color{black}$(2^k,2^{m-k-1},2^{m}),$ $m\geq 2,$ and $k\leq m-1$.} &  Direct &  Optimal for binary case\\ \midrule
        
        \cite{Th} & PS\footnotemark[1] & \makecell[l]{$(2(2n+1),4k+1,4(2k+1)(2n+1)),$ \\ $ n\geq 1,\ k\geq 1$} & Indirect & Neither optimal nor almost-optimal\\\hline

        \cite{LiGuPa} & GBF & $(2^{k+1},2^n,2^{n+k+2}),k\geq 0,n\geq 1$ & Direct & Optimal for binary case\\ \hline
        \cite{ZhDz} & PNLF\footnotemark[2] & $(p,p-1,p^2)$, $p$ is an odd prime &  Direct & Optimal\\ \hline
        \cite{ZhPa} & GBnF\footnotemark[3] & $(N,N-1,N^2)$, $N$ is positive integer. &  Direct & Optimal.\\
        \hline
        
        Theorem 1 & EBF & \makecell[l]{${(p^t,(p-1)p^n,p^{n+t+1})},\ n,t\geq 1,$ \\$ p$ is prime and $t\leq n$} & Direct & \makecell[l]{Optimal for binary case else\\ asymptotically optimal}\\\hline
    \end{tabular}\label{table2}
    \footnotetext[1]{Perfect Sequences.}
    \footnotetext[2]{Perfect Non-linear functions.}
    \footnotetext[3]{Generalised Bent functions.}
\end{sidewaystable}
    

\subsection{Comparison with Traditional ZCZ Sequences in \cite{ZhDz} and \cite{ZhPa}}
In \cite{ZhDz} and \cite{ZhPa}, parameters of constructed ZCZ sequence sets appear to be $(p,p,p^2)$ and $(N, N, N^2)$ respectively, where $p$ is an odd prime, and $N$ is a positive integer. In both constructions, the length parameter is fixed, i.e., none of them can produce a ZCZ sequence of length $p^3$. Hence flexibility in the length parameter is not much here.  Also, once  the length parameter is fixed, it is not possible to play with ZCZ width and set size, i.e., fixing of length parameter fixes ZCZ width and set size. On the other hand, in our proposed construction, the length parameter is $p^{n+t+1}$, i.e., the length in the form of $p^k$, where $k\geq 3$ can be achieved. Also, for the fixed value of $n+t,~n\geq1$, and $t\leq n$, more than one combinations of $n$ and $t$ exist. Hence, the ZCZ width and set size can be altered with the fixed length.
Thus our construction is more flexible, as it can generate ZCZ sequences of prime-powers length, and also, it is direct.

\subsection{Comparison with Traditional ZCZ Sequences in \cite{pai2022optimal}}

Recently, as a generalization of the proposed work, a new construction of $(b^k,(b-1)b^{m-k-1}+(b-2)b^{m-k-2},b^m)$-ZCZ sequences using EGBFs has been reported in \cite{pai2022optimal}. The parameter of ZCZ sequence set in \cite{pai2022optimal} is better than ours. But, based on the record of the publication date of \cite{pai2022optimal}, the proposed work \cite{NiSu} has appeared online in Arxive earlier date.

\subsection{Comparison with SNC-ZCZ sequences in \cite{li2022spectrally}}
In \cite{li2022spectrally}, the authors proposed the construction of $(N,N,(N+e))$-SNC ZCZ sequence set, where $N$ and $e$ are positive integers. in our construction, we proposed $(p^t,(p-1)p^n,p^{n+t+1})$-SNC-ZCZ sequence set. In comparison to \cite{li2022spectrally}, the proposed construction offers following benefits
\begin{itemize}
    \item The ZCZ and the number of sequences are dependent (as they are the same) in \cite{li2022spectrally}, while we don't have this type of constraint from the proposed construction, i.e., they are independent. 
    \item The alphabet size of SNC-ZCZ sequences in \cite{li2022spectrally} is $lcm(N, N+e)$, which is larger than the alphabet size of our proposed construction, i.e., $p+1$. For example, the SNC-ZCZ of length $27$ that can be constructed from \cite{li2022spectrally} has parameter $(3,3,27)$ and alphabet size $9$, while for the same length, our construction can produce $(3,6,27)$-SNC-ZCZ with alphabet size $3$.
\end{itemize}

\subsection{Comparison with SNC-ZCZ sequences in \cite{ye2022new}}
In \cite{ye2022new}, authors presented construction of SNC-ZCZ sequences with parameter $(\Tilde{F}(N),N+1,N(N+1))$, where for each positive integer $N\geq 2$, $\Tilde{F}(N)$ denotes the maximum number of rows such that an $\Tilde{F}(N)\times N$ circular Florentine rectangle (CFR) exist. In comparison to  \cite{ye2022new}, our construction gives following benefits,
\begin{itemize}
    \item The construction in \cite{ye2022new} is heavily depends on the existence of CFR, which makes this construction very restricted. While the proposed construction is based on EBFs which are well known in the existing literature.
    \item The parameters of SNC-ZCZ in \cite{ye2022new} are interdependent (as they are all dependent on $N$), which restrict the flexibility of parameters. On the other hand, our construction is much more flexible. 
\end{itemize}
The construction in \cite{ye2022new} is heavily depends on the existence of CFR. which makes this construction very limited
\section{Conclusion}\label{sec13}

In this paper, we introduce a construction method for SNC-ZCZ sequences of prime-power length based on EBFs having flexible support. The proposed SNC-ZCZ sequences are asymptotically optimal. Additionally, we establish a relationship between traditional ZCZ sequences and GRM codes.

\section*{Declaration}
\textbf{Conflict of interest} The authors report no conflict of interest.\\
\textbf{Data availability} Not Applicable.
\bibliography{main}
\end{document}